\documentclass[10pt, twocolumn, a4paper]{article}

\usepackage{hyperref}
    \hypersetup{colorlinks=true, linkcolor=black, citecolor=black, urlcolor=black}
\usepackage[top=2cm, bottom=2cm, left=1.5cm, right=1.5cm]{geometry}
\usepackage{fancyhdr}
\usepackage{microtype}
\usepackage{titlesec}
    \titleformat{\title}{\large\bfseries}{}{}{}
    \titleformat{\section}{\normalfont\bfseries}{\thesection}{0.5em}{}
    \titleformat{\subsection}{\normalfont\it}{\thesubsection}{0.5em}{}
    \titleformat{\subsubsection}{\normalfont\normalsize\it}{\thesubsubsection}{0.5em}{}
    \titleformat{\paragraph}[runin]{\normalfont\bfseries}{\theparagraph}{0.5em}{}
    \titleformat{\subparagraph}[runin]{\normalfont\normalsize\it}{\thesubparagraph}{0.5em}{}
\usepackage[font=small,labelfont=bf,labelsep=space]{caption}
\usepackage[shortlabels, inline]{enumitem} 
\usepackage{subfig}
\usepackage{natbib}

\usepackage{booktabs}
\usepackage[retainorgcmds]{IEEEtrantools}

\usepackage{comment}
\usepackage{etoolbox}
    \newbool{PRECOMPILED}
    \booltrue{PRECOMPILED}

\usepackage{amssymb}
\usepackage{amsmath}
\usepackage{amsthm}
\usepackage{bbm} 
	\newcommand{\1}{\mathbbm{1}}
    \providecommand{\1}{\mathds{1}}

\usepackage{tikz}
    \definecolor{myblue}{HTML}{000099}
    \definecolor{mylightblue}{HTML}{CCCCFF}
    \definecolor{myorange}{HTML}{FF6600}
    \definecolor{myred}{HTML}{990000}
    \definecolor{mygreen}{HTML}{009900}
    \definecolor{mylightgray}{HTML}{E0E0E0}
    \usetikzlibrary{decorations.pathmorphing}
	
    \tikzstyle{vertex}=[circle, draw, fill=white, inner sep=0pt, minimum width=1ex]
    \tikzset{every picture/.append style={baseline,scale=1.1}}
    \tikzstyle{cut-edge}=[dotted]
    \tikzstyle{base}=[line width=1.2pt]
    \tikzstyle{lifted}=[blue]

\usepackage{sidecap}

\providecommand{\R}{\mathbb{R}}
\providecommand{\Z}{\mathbb{Z}}

\providecommand{\Q}{\mathbb{Q}}

\theoremstyle{plain}
\newtheorem{thm}{Theorem}
\newtheorem{lem}{Lemma}

\newtheorem{cor}{Corollary}

\theoremstyle{definition}
\newtheorem{defn}{Definition}
\newtheorem{fact}{Fact}
\newtheorem{rem}{Remark}

\DeclareMathOperator{\dist}{dist}
\DeclareMathOperator{\conv}{conv}
\DeclareMathOperator{\lin}{lin}

\newcommand{\np}{\textsc{np}}
\newcommand{\mybinom}[2]{\raisebox{0.1ex}{\scalebox{0.8}{$\tbinom{#1}{#2}$}}}
\newcommand{\treeproblem}{tree partition problem}
\newcommand{\pathproblem}{path partition problem}

\makeatletter
\def\moverlay{\mathpalette\mov@rlay}
\def\mov@rlay#1#2{\leavevmode\vtop{%
   \baselineskip\z@skip \lineskiplimit-\maxdimen
   \ialign{\hfil$\m@th#1##$\hfil\cr#2\crcr}}}
\newcommand{\charfusion}[3][\mathord]{
    #1{\ifx#1\mathop\vphantom{#2}\fi
        \mathpalette\mov@rlay{#2\cr#3}
      }
    \ifx#1\mathop\expandafter\displaylimits\fi}
\makeatother
\newcommand{\cupdot}{\charfusion[\mathbin]{\cup}{\cdot}}

\begin{document}

\title{\bf\Large Decomposition of Trees and Paths via Correlation}
\author{Jan-Hendrik Lange \\ \small{\texttt{jlange@mpi-inf.mpg.de}} \and Bjoern Andres \\ \small{\texttt{andres@mpi-inf.mpg.de}}}
\date{Max Planck Institute for Informatics \\ Saarland Informatics Campus $\cdot$ Saarbr\"ucken, Germany\\}

\twocolumn[
\begin{@twocolumnfalse}
\maketitle
\begin{abstract}
We study the problem of decomposing (clustering) a tree with respect to costs attributed to pairs of nodes, so as to minimize the sum of costs for those pairs of nodes that are in the same component (cluster).
For the general case and for the special case of the tree being a star, we show that the problem is \np-hard.
For the special case of the tree being a path, this problem is known to be polynomial time solvable. We characterize several classes of facets of the combinatorial polytope associated with a formulation of this clustering problem in terms of lifted multicuts. In particular, our results yield a complete totally dual integral (TDI) description of the lifted multicut polytope for paths, which establishes a connection to the combinatorial properties of alternative formulations such as set partitioning.
\end{abstract}
\vspace{7ex}
\end{@twocolumnfalse}
]



\section{Introduction}

We study the problem of decomposing (clustering) a tree with respect to costs attributed to pairs of nodes, so as to minimize the sum of costs for those pairs of nodes that are in the same component (cluster).
This \emph{\treeproblem{}} is stated rigorously in Def.~\ref{def:problem-main}.
One instance and its solution are depicted in Fig.~\ref{fig:tree-decomposition-costs}.

On the one hand, the \treeproblem{} lacks a complexity of clustering problems for general graphs:
Its feasible set, the set of all decompositions of a tree, is trivial to characterize.
For example, the decompositions of any tree $T = (V, E)$ relate one-to-one to the binary edge labelings $y \in \{0,1\}^E$ that indicate whether neighboring nodes $\{u, v\} \in E$ are in the same component ($y_{uv} = 1$) or distinct components ($y_{uv} = 0$).

On the other hand, the \treeproblem{} exhibits a complexity not commonly found in clustering problems for general graphs:
Its objective function has a multilinear form in the coordinates of $y$ (in the standard unit basis) whose polynomial degree is the length of the longest path in $T$.
This form arises from the fact that any two distinct nodes $\{u, v\} \in \tbinom{V}{2}$ are in the same component iff $y_e = 1$ for all edges $e$ on the unique path $P_{uv}$ in $T$ from $u$ to $v$, that is, iff $\prod_{e \in P_{uv}} y_e = 1$. 
Hence, the \treeproblem{} can be stated rigorously as follows:

\begin{defn}
\label{def:problem-main}
For any (finite, simple, undirected) tree $T = (V, E)$ and 
any $c \colon \tbinom{V}{2} \to \mathbb{R}$,
the optimization problem \eqref{eq:problem-pbo} is called the instance of the \emph{\treeproblem{}} w.r.t.~$T$ and $c$.
If $T$ is a path, it is also called the instance of the \emph{\pathproblem{}} w.r.t.~$T$ and $c$.
\begin{align}
\min_{y \in \{0,1\}^E}
\sum_{\{u,v\} \in \mybinom{V}{2}}
    c_{uv} 
\prod_{e \in P_{uv}}
    y_e
\label{eq:problem-pbo}
\end{align}
\end{defn}

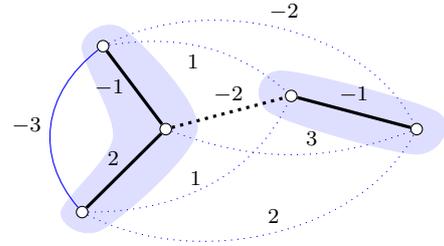
\begin{figure}
\centering
\begin{tikzpicture}[shift={(0,-0.0)},xscale=2.5,yscale=2.0]
\draw[draw=mylightblue!65, fill=mylightblue!65] plot[smooth cycle, tension=0.5] coordinates
    {(-0.1, 0) (0.15, 0.45) (0, 1) (0.2, 1.1) (0.55, 0.5) (0.1, -0.1)};
\draw[draw=mylightblue!65, fill=mylightblue!65] plot[smooth cycle, tension=0.5] coordinates
    {(0.9, 0.6) (0.95, 0.85) (1.7, 0.6) (1.6, 0.35)};

\node (0) [style=vertex,fill=white,label=below:{\scriptsize}] at (0, 0) {};
\node (1) [style=vertex,fill=white,label=below:{\scriptsize}] at (0.1, 1.0) {};
\node (2) [style=vertex,fill=white,label=below:{\scriptsize}] at (0.4, 0.5) {};
\node (3) [style=vertex,fill=white,label=above:{\scriptsize}] at (1.0, 0.7) {};
\node (4) [style=vertex,fill=white,label=below:{\scriptsize}] at (1.6, 0.5) {};

\draw (0) edge[base] node[left=4,above=-2] {\footnotesize $2$} (2);
\draw (1) edge[base] node[left] {\footnotesize $-1$} (2);
\draw (2) edge[cut-edge,base] node[above] {\footnotesize $-2$} (3);
\draw (3) edge[base] node[above] {\footnotesize $-1$} (4);

\draw[lifted] (0) edge[bend left=40] node[left,black] {\footnotesize $-3$} (1);
\draw[cut-edge,lifted] (1) edge[bend left=30] node[left=5, below=-1,black] {\footnotesize $1$} (3);
\draw[cut-edge,lifted] (0) edge[bend right=35] node[left=6,above=-4,black] {\footnotesize $1$} (3);
\draw[cut-edge,lifted] (0) edge[bend right=50] node[above,black] {\footnotesize $2$} (4);
\draw[cut-edge,lifted] (1) edge[bend left=50] node[above,black] {\footnotesize $-2$} (4);
\draw[cut-edge,lifted] (2) edge[bend right=25] node[above=6,right=2,black] {\footnotesize $3$} (4);
\end{tikzpicture}
\caption{Depicted above is an instance of the \treeproblem{} with costs associated to pairs of nodes. 
Edges of the tree are depicted as straight bold lines (dotted or solid).
Additional pairs of nodes are marked by thin blue curves (dotted or solid).
The solution decomposes the tree into two components, depicted with a shaded background, leaving some pairs of nodes in the same component (solid lines and curves) and some pairs of lines in distinct components (dotted lines and curves).}
\label{fig:tree-decomposition-costs}
%
%
\end{figure}

\vspace{-3ex}

\subsection{Contribution} 

In Section~\ref{sec:complexity}, we survey several representations of the \treeproblem{} and discuss its computational complexity.
For the general case and for the special case of the tree being a star, we show that the problem is \np-hard. The \pathproblem{}, on the other hand, is known to be polynomial time solvable \citep{Kernighan1971}.

In Section~\ref{sec:polyhedral}, we study the lifted multicut polytope \citep{hornakova-2017} associated with the lifted multicut representation of the \treeproblem{}. We characterize several classes of facets for the general case, which in particular yield a complete totally dual integral description (TDI) of the lifted multicut polytope for paths. Our results relate the geometry of lifted multicuts to the combinatorial properties of the tree and \pathproblem{}.

\subsection{Notation}

Throughout the paper, we consider $T = (V,E)$ to be a (finite, simple, undirected) tree with $n = \lvert E \rvert$ edges. Let $m = \lvert \mybinom{V}{2} \rvert$ be the number of distinct pairs of nodes in $V$.
For any $\{u,v\} \in \tbinom{V}{2}$, we denote by $P_{uv}$ the unique path in $T$ connecting $u$ and $v$. 
Moreover, we denote by $\dist(u,v)$ the distance of $u$ and $v$ in $T$, i.e.\ the length of $P_{uv}$. 
To ease notation, we take the liberty of identifying any graph $G = (V_G, E_G)$ with its edge set $E_G$.
Hence, we may write $e \in G$ instead of $e \in E_G$. 
For any $k \in \mathbb{N}$, we denote by $[k] := \{1, \dotsc, k\}$ the set of integers from $1$ through $k$.


\section{Related Work}

\subsection{Pseudo-Boolean Optimization}
The optimization of pseudo-Boolean functions plays an important role in machine learning, for instance, in MAP-inference for computer vision models. The general problem can be reduced to the quadratic case \citep{Boros2002, Boros2012}, which is responsible for the \np-hardness of the problem. The combinatorial polytope associated with the linearization of quadratic pseudo-Boolean functions was extensively studied by \citet{Padberg1989}. Computational approaches to pseudo-Boolean optimization based on the roof duality bound \citep{Hammer1984,Kahl2012} have been quite successful in practice \citep{Rother2007}. 

\subsection{Correlation Clustering and Multicut Polytopes}
Decompositions of graphs into an unknown number of components based on (dis-)similarities between \emph{neighboring} nodes is referred to as weighted correlation clustering. It has been studied for the complete graph \citep{Bansal2004} as well as general graphs \citep{Demaine2006}. As any decomposition of a graph is characterized by the mathematical notion of a multicut, the study of multicuts is closely related to correlation clustering. The combinatorial polytopes associated with the multicuts of a graph have been studied, among others, most notably by \citet{Groetschel1989, Chopra1993, Deza1997}. The more general case where correlations between \emph{non-neighboring} nodes are taken into account as well has been introduced by \citet{hornakova-2017}.

\subsection{Sequential Set Partitioning}
Set partitioning problems where the elements are assumed to adhere to a linear order have been studied by \citet{Kernighan1971}, who devises a dynamic programming algorithm that essentially solves a shortest path problem in a directed acyclic graph. The corresponding integer linear programming formulation admits a totally unimodular constraint matrix \citep{Joseph1997}.


\section{Representations and Complexity} 
\label{sec:complexity}

\subsection{Sparse Pseudo-Boolean Functions}

For any $n \in \mathbb{N}$, any $f \colon \{0,1\}^n \to \R$ is called an $n$-variate \emph{pseudo-Boolean} function (PBF). Any $n$-variate PBF $f$ has a unique multi-linear polynomial form\footnote{In line with the literature on pseudo-Boolean optimization, we call the form \eqref{eq:multi-linear-poly} \emph{multi-linear} despite its constant term $c_\emptyset$.}:
\begin{align}
f(y) = \sum_{I \subseteq [n]} c_I \prod_{i \in I} y_i, 
\qquad \textnormal{with} \qquad
c \colon 2^{[n]} \to \R.
\label{eq:multi-linear-poly}
\end{align}
We observe that the objective functions of the \treeproblem{} form a class of sparse PBFs whose set $\{I \subseteq [n] \colon c_I \neq 0\}$ of non-zero coefficients is constrained by a tree. This class of functions is formalized in the following definition.

\begin{defn}
An $n$-variate PBF is called \emph{tree-sparse} w.r.t.~a tree $T = (V, [n])$ iff its multi-linear polynomial form \eqref{eq:multi-linear-poly} is such that, for every $I \subseteq [n]$ with $c_I \neq 0$, the set $I$ induces a path in $T$.
The function is called \emph{path-sparse} w.r.t.~a path $P = (V, [n])$ iff it is tree-sparse w.r.t.~$P$.
\end{defn}

\subsection{Set Partitioning}

For any connected subgraph (subtree) $S \subseteq T$, introduce a binary variable $\lambda_S \in \{0,1\}$. By defining the cost $d_S = \sum_{u,v \in V_S} c_{uv}$ for each component $S$, problem \eqref{eq:problem-pbo} can be reformulated as the set partitioning problem
\begin{align}
\min \quad & d^\top \lambda \label{eq:problem-set-partition} \\
\text{s.t.} \quad & \sum_{S : u \in V_S} \lambda_S = 1, \quad \forall u \in V \nonumber \\
& \lambda_S \geq 0, \; \lambda_S \in \Z. \nonumber
\end{align}
Here, the costs $d_S$ account for all pairs of nodes within component $S$ and the constraints ensure that every node is contained in exactly one component of the partition. Note that the number of variables $\lambda_S$ is exponential in the number of leaves of $T$.

\begin{lem}
The vector $y \in \{0,1\}^E$ is a solution of problem \eqref{eq:problem-pbo} w.r.t.\ the tree $T=(V,E)$ and costs $c \colon \tbinom{V}{2} \to \R$ iff the vector $\lambda$ defined by
\begin{align*}
\lambda_S = 1 \iff \forall e \in S : y_e = 1
\end{align*}
for every subtree $S \subseteq T$, is a solution of problem \eqref{eq:problem-set-partition} w.r.t.\ the costs $d_S = \sum_{u,v \in V_S} c_{uv}$.
\end{lem}

\begin{proof}
We set $c_{uu} = 0$ for all $u \in V$. Let $S_u$ denote the component of the partition that contains $u \in V$. The claim follows from
\begin{align*}
\sum_{\{u,v\} \in \mybinom{V}{2}} c_{uv} \prod_{e \in P_{uv}} y_e = & \sum_{u \in V} \sum_{v \in V_{S_u}} c_{uv} \prod_{e \in P_{uv}} y_e \\
= & \sum_{u \in V} \sum_{v \in V_{S_u}} c_{uv} \\
= & \sum_{u \in V} \sum_{\substack{\text{subtree} \\ S \subseteq T : u \in V_S}} \lambda_S \sum_{v \in V_S} c_{uv} \\
= & \sum_{\substack{\text{subtree} \\ S \subseteq T}} \lambda_S \sum_{u,v \in V_S} c_{uv}. \qedhere
\end{align*}
\end{proof}

\subsection{Lifted Multicuts}

In this section, we identify the \treeproblem{} (Def.~\ref{def:problem-main}) as a special case of the minimum cost lifted multicut problem \citep{hornakova-2017} where the underlying graph is a tree. Essentially, this representation is a linearization of the \treeproblem{} with reversed binary encoding.
 
The \emph{minimum cost lifted multicut problem} w.r.t.\ the tree $T=(V,E)$ and costs $c \colon \tbinom{V}{2} \to \R$ is the combinatorial optimization problem
\begin{align}
\min_{x \in X_T} \sum_{\{u,v\} \in \mybinom{V}{2}} c_{uv} \, x_{uv}, \label{eq:problem-lmp}
\end{align}
where $X_T$ is defined as
\begin{align}
X_{T} & = \Big \{ x \in \{0,1\}^m \; \big| \nonumber \\
& x_{uv} \leq \sum_{e \in P_{uv}} x_e && \hspace{-1em} \forall u,v \in V, \dist(u,v) \geq 2, \label{eq:tree-lmp-path} \\
& x_e \leq x_{uv} && \hspace{-1em} \forall u,v \in V, \dist(u,v) \geq 2, \; \forall e \in P_{uv} \label{eq:tree-lmp-cut} \Big \}.
\end{align}
Problem \eqref{eq:problem-lmp} is the special case of the minimum cost lifted multicut problem as presented by \citet{hornakova-2017} for the specific choice $G = T$ and $G' = (V, \mybinom{V}{2})$, i.e.\ the tree $T$ lifted to the complete graph on $V$ (cf.\ appendix \ref{sec:lifted-multicuts}). Lemma \ref{lem:tree-lmp} states that the \treeproblem{} can be reformulated as a minimum cost lifted multicut problem.
\begin{lem} \label{lem:tree-lmp}
The vector $y \in \{0,1\}^n$ is a solution of problem \eqref{eq:problem-pbo} w.r.t.\ the tree $T=(V,E)$ and costs $c \colon \mybinom{V}{2} \to \R$ 
iff
the unique $x \in X_T$ such that $x_e = 1 - y_e$ for all $e \in E$ is a solution of problem \eqref{eq:problem-lmp} w.r.t.\ $T$ and the cost function $-c$.
\end{lem}

\begin{proof}
For any distinct pair of nodes $u,v \in V$, we introduce a binary variable $x_{uv} \in \{0,1\}$ via
\begin{align}
x_{uv} = 1 - \prod_{e \in P_{uv}} y_e
\end{align}
which implies
\begin{align}
x_{uv} = 0
& \quad \iff \quad 
\forall e \in P_{uv} \colon \ y_e = 1 \nonumber \\
& \quad \iff \quad 
\forall e \in P_{uv} \colon \ x_e = 0. \label{eq:vars-transform}
\end{align}
Therefore, we can reformulate problem \eqref{eq:problem-pbo} in terms of the variables $x_{uv}$ by transforming the objective function according to
\begin{align}
c_{uv} \prod_{e \in P_{uv}} y_e & = -c_{uv} \big(1 - \prod_{e \in P_{uv}} y_e \big) + c_{uv} \nonumber \\
& = -c_{uv} \, x_{uv} + c_{uv}.
\end{align}
This leads to the combinatorial optimization problem
\begin{align}
\min_{x \in X_T} \sum_{\{u,v\} \in \mybinom{V}{2}} -c_{uv} \, x_{uv} + c_{uv}, \label{eq:tree-lmp-problem}
\end{align}
where $X_T \subseteq \{0,1\}^m$ captures the relationship \eqref{eq:vars-transform}.
\end{proof}

Note that, since the objective function of \eqref{eq:problem-lmp} is linear, we can replace $X_T$ by its convex hull
\begin{align}
\Xi_T = \conv X_T,
\end{align}
which is called the \emph{lifted multicut polytope} w.r.t.\ $T$. This motivates the study of the structure of $\Xi_T$ in Section \ref{sec:polyhedral}.
We refer to \eqref{eq:tree-lmp-path} and \eqref{eq:tree-lmp-cut} as \emph{path} and \emph{cut} inequalities, respectively, and write $\Theta^0_T$ for the naive linear relaxation of $\Xi_T$, i.e.\ the set of vectors $x \in [0,1]^m$ that satisfy \eqref{eq:tree-lmp-path} and \eqref{eq:tree-lmp-cut}.

\subsection{Complexity}

We now discuss the computational complexity of the \treeproblem{} and the \pathproblem{}.
In Lemma~\ref{lem:np-hardness}, we show that the \treeproblem{} is \np-hard, even if the tree is a star.
On the contrary, the \pathproblem{} is polynomial time solvable \citep{Kernighan1971}.

\begin{lem} 
\label{lem:np-hardness}
The \treeproblem{} is \np-hard.
It remains \np-hard if $T$ is a star.
\end{lem}

\begin{proof}
For a tree-sparse pseudo-Boolean function defined on a star with $n$ leaves, problem \eqref{eq:problem-pbo} is equivalent to the unconstrained binary quadratic program with $n$ variables, which is well-known to be \np-hard.
\end{proof}


\section{Polyhedral Geometry} 
\label{sec:polyhedral}

\begin{figure*}
\subfloat[]{%
\begin{tikzpicture}[shift={(0,-0.35)},scale=1.3]
\node (0) [style=vertex,fill=white,label=below:{\scriptsize $0$}] at (0, 0.2) {};
\node (1) [style=vertex,fill=white,label=below:{\scriptsize $1$}] at (0.6, 0.6) {};
\node (2) [style=vertex,fill=white,label=below:{\scriptsize $2$}] at (1.3, 1.0) {};
\node (3) [style=vertex,fill=white,label=right:{\scriptsize $3$}] at (0.7, 1.2) {};
\node (4) [style=vertex,fill=white,label=below:{\scriptsize $4$}] at (0.1, 1.5) {};
\node (5) [style=vertex,fill=white,label=below:{\scriptsize $5$}] at (1.9, 0.7) {};
\node (6) [style=vertex,fill=white,label=below:{\scriptsize $6$}] at (1.7, 1.5) {};

\draw (0) edge[base] node[right] {} (1);
\draw (1) edge[base] node[right] {} (2);
\draw (1) edge[base] node[right] {} (3);
\draw (3) edge[base] node[right] {} (4);
\draw (2) edge[base] node[right] {} (5);
\draw (2) edge[base] node[right] {} (6);

\draw[-] (4) edge[bend left=50,lifted] node[below,black] {$x_{46}$} (6);
\draw[-] (0) edge[bend left=30,lifted] node[above] {} (3);
\draw[-] (0) edge[bend right=40,lifted] node[above,black] {$x_{05}$} (5);
\end{tikzpicture}
\label{fig:graphs-tree}
}
\hfill
\subfloat[]{%
\begin{tikzpicture}[scale=1]
\node (1) [style=vertex,fill=white,label=left:{\scriptsize $1$}] at (0, 0) {};
\node (2) [style=vertex,fill=white,label=right:{\scriptsize $3$}] at (2, 0) {};
\node (3) [style=vertex,fill=white,label=above:{\scriptsize $2$}] at (1, 1.73) {};
\node (4) [style=vertex,fill=white,label=below:{\scriptsize $0$}] at (1, 0.575) {};

\draw[-] (1) edge[bend right=50,lifted] node[above,black] {$x_{13}$} (2);
\draw[-] (2) edge[bend right=50,lifted] node[right,black] {$x_{23}$} (3);
\draw[-] (1) edge[bend left=50,lifted] node[left,black] {$x_{12}$} (3);
\draw (3) edge[base] (4);
\draw (1) edge[base] node[above=4] {$x_{01}$} (4);
\draw (2) edge[base] (4);
\end{tikzpicture}
\label{fig:graphs-star}
}
\hfill
\subfloat[]{%
\begin{tikzpicture}[shift={(0,0.25)},scale=1]
\node (0) [style=vertex,fill=white,label=below:{\scriptsize $0$}] at (0, 0) {};
\node (1) [style=vertex,fill=white,label=below:{\scriptsize $1$}] at (1, 0) {};
\node (2) [style=vertex,fill=white,label=below:{\scriptsize $2$}] at (2, 0) {};
\node (3) [style=vertex,fill=white,label=below:{\scriptsize $3$}] at (3, 0) {};

\draw (0) edge[base] node[right] {} (1);
\draw (1) edge[base] node[above] {$x_{12}$} (2);
\draw (2) edge[base] node[above] {} (3);
\draw[-] (0) edge[bend left=60,lifted] node[above=3,right=10,black] {$x_{02}$} (2);
\draw[-] (1) edge[bend right=55,lifted] node[below,black] {$x_{13}$} (3);
\draw[-] (0) edge[bend left=80,lifted] node[above,black] {$x_{03}$} (3);
\end{tikzpicture}
\label{fig:graphs-path}
}
\hfill
\subfloat[]{%
\begin{tikzpicture}[scale=1.3]
\draw[base] plot [smooth] coordinates {(1,0) (1.2,0.1) (1.45,0.05)};
\draw[base] plot [smooth] coordinates {(2.0,-0.05) (2.3,-0.1) (2.5,0)};
\node (dots) at (1.75,0) {$\dotso$};
\node (0) [style=vertex,fill=white,label=left:{\scriptsize $u$}] at (0.5, 0.5) {};
\node (1) [style=vertex,fill=white,label=below:{\scriptsize $\vec{u}(v)$}] at (1, 0) {};
\node (2) [style=vertex,fill=white,label=below:{\scriptsize $v$}] at (2.5, 0) {};

\draw (0) edge[base] node[right] {} (1);
\draw[-] (0) edge[bend left=70,lifted] node[above=3,right=10,black] {} (2);
\draw[-] (1) edge[bend left=60,lifted] node[below,black] {} (2);
\end{tikzpicture}
\label{fig:graphs-triangle}
}
\caption{\ref{fig:graphs-tree}: A tree with additional lifted edges (thin blue curves) between non-neighboring nodes corresponding to long range variables. \ref{fig:graphs-star}: A star with three leaves lifted to the complete graph. \ref{fig:graphs-path}: A path of length $3$ lifted to the complete graph. \ref{fig:graphs-triangle}: The node $\vec{u}(v)$ is the first node on the path $P_{uv}$.}
\label{fig:graphs}
\end{figure*}
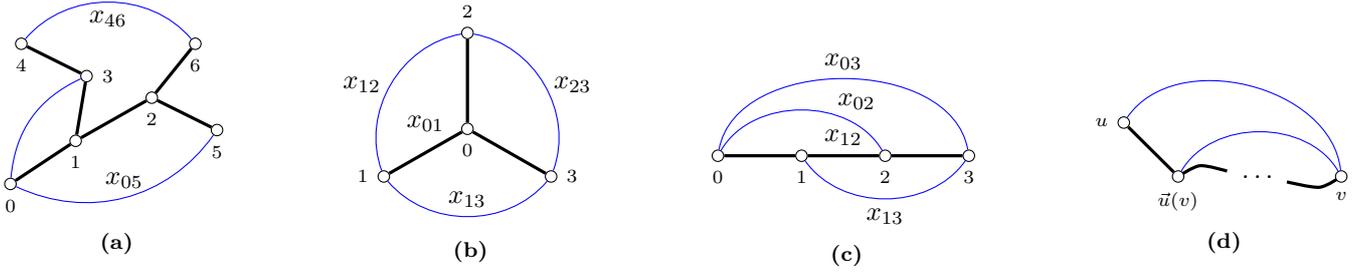

In this section, we establish polyhedral results for the lifted multicut polytope $\Xi_T$ for trees. We characterize all trivial facets and offer a tighter outer relaxation of $\Xi_T$. In Section~\ref{sec:paths}, we show that our results yield a complete totally dual integral (TDI) description of the lifted multicut polytope for paths. This result relates the combinatorial properties of the sequential set partitioning problem to the geometry of the minimum cost lifted multicut problem for paths. The rather technical proofs for our claims are deferred to appendix \ref{sec:proofs}.

\subsection{Canonical Outer Relaxation} \label{sec:tree-relax}
We give another simple relaxation of $\Xi_T$ that is at least as tight as the naive linear relaxation $\Theta^0_T$. To this end, define $\vec{u}(v)$ to be the first node on the path $P_{uv}$ that is different from both $u$ and $v$ (cf.\ Fig.\ \ref{fig:graphs-triangle}) and consider the polytope
\begin{align}
\Theta^1_T & = \Big \{ x \in [0,1]^m \; \big | \nonumber \\
& x_{uv} \leq x_{u,\vec{u}(v)} + x_{\vec{u}(v),v} && \forall u,v \in V, \; \dist(u,v) \geq 2, \label{eq:tree-path}  \\
& x_{\vec{u}(v),v} \leq x_{uv} && \forall u,v \in V, \; \dist(u,v) \geq 2 \label{eq:tree-cut} \Big \}.
\end{align}

This description is canonical in the sense that it only considers a quadratic number of node triplets, namely those which feature two neighboring nodes and an arbitrary third node. The following lemma states that $\Theta^1_T$ is indeed an outer relaxation of $\Xi_T$ which is at least as tight as $\Theta^0_T$.

\begin{lem} \label{lem:tree-polytope}
It holds that $\Xi_T \subseteq \Theta^1_T \subseteq \Theta^0_T$.
\end{lem}

\subsection{Trivial Facets} 
\label{sec:tree-facets}

It can be characterized exactly which inequalities of \eqref{eq:tree-path}, \eqref{eq:tree-cut} and $0 \leq x_{uv} \leq 1$ define facets of $\Xi_T$. These elementary results are closely related to the more general study in \citep{hornakova-2017}.

\begin{lem} \label{lem:tree-path-facets}
The inequality \eqref{eq:tree-path} defines a facet of $\Xi_T$ if and only if $\dist(u,v) = 2$.
\end{lem}

\begin{lem} \label{lem:tree-cut-facets}
The inequality \eqref{eq:tree-cut} defines a facet of $\Xi_T$ if and only if $v$ is a leaf of $T$.
\end{lem}

\begin{lem} \label{lem:tree-box-facets}
For any $u,v \in V,\; u \neq v$, the inequality $x_{uv} \leq 1$ defines a facet of $\Xi_T$ if and only if both $u$ and $v$ are leaves of $T$. Moreover, none of the inequalities $0 \leq x_{uv}$ define facets of $\Xi_T$.
\end{lem}

\subsection{Non-Trivial Facets}

We further present a large class of non-trivial facets of $\Xi_T$. Consider the set of inequalities given by
\begin{align}
x_{uv} + x_{\vec{u}(v),\vec{v}(u)} & \leq x_{u,\vec{v}(u)} + x_{\vec{u}(v),v} \label{eq:tree-square} \\
& \forall u,v \in V,\; \dist(u,v) \geq 3. \nonumber
\end{align}
We shall refer to \eqref{eq:tree-square} as \emph{square} inequalities. As an example consider the graph depicted in Fig.\ \ref{fig:graphs-path} and $u=0$, $v=3$. Any square inequality is valid and facet-defining for $\Xi_T$.

\begin{lem} \label{lem:tree-square-valid}
The inequalities \eqref{eq:tree-square} are valid for $\Xi_T$.
\end{lem}

\begin{lem} \label{lem:tree-square-facets}
The inequalities \eqref{eq:tree-square} define facets of $\Xi_T$.
\end{lem}

\subsection{Facets Related to the Boolean Quadric Polytope}

The argument to show Lemma \ref{lem:np-hardness} also implies that the lifted multicut polytope for a star with $n$ leaves is isomorphic to the \emph{Boolean Quadric Polytope} \citep{Padberg1989} of order $n$. Therefore, insights about the facial structure of the Boolean quadric polytope can be used to derive valid inequalities and facets for the more general polytope $\Xi_T$. In fact, any tree-sparse pseudo-Boolean function repeatedly admits substructures similar to the quadratic case. As this connection is beyond the scope of our paper, we only illustrate it with the following example. Consider the star graph depicted in Fig.\ \ref{fig:graphs-star}. It is known from \citep{Padberg1989} that the ``triangle'' inequality
\begin{align}
x_{01} + x_{23} \leq x_{12} + x_{13} \label{eq:triangle}
\end{align}
is valid and facet-defining for $\Xi_T$. Observe that, similarly, for the tree in Fig.\ \ref{fig:graphs-tree} 
\begin{align}
x_{02} + x_{56} \leq x_{05} + x_{06}
\end{align} 
holds true, but in this case $x_{02}, x_{05}$ and $x_{06}$ correspond to cubic instead of quadratic terms of the associated pseudo-Boolean function.

\subsection{Complete TDI Description for Paths}
\label{sec:paths}

Suppose the node set $V = \{0, \dotsc, n\}$ is linearly ordered and consider the path $P = (V,E)$, where the edge set is given by $E = \big\{ \{i,i+1\} \mid i \in \{0, \dotsc, n-1\} \big \}$. We show that the facets arising from \eqref{eq:tree-path} - \eqref{eq:tree-square} yield a complete description of $\Xi_P$.
For this section, we consider only paths of length $n \geq 2$, since for $n=1$ the polytope $\Xi_P = [0,1]$ is trivial. Let $\Theta^{\mathrm{Path}}_P$ be the polytope defined by
\begin{align}
\Theta^{\mathrm{Path}}_P = \Big \{ & x \in \R^m \; \big | \nonumber \\
 & x_{0n} \leq 1, \label{eq:path-box} \\
& x_{in} \leq x_{i-1,n} & \hspace{-7em} \forall i \in \{1, \dotsc, n-1\}, \label{eq:path-cut-right} \\
& x_{0i} \leq x_{0,i+1} & \hspace{-7em} \forall i \in \{1, \dotsc, n-1\}, \label{eq:path-cut-left} \\
& x_{i-1,i+1} \leq x_{i-1,i} + x_{i,i+1} & \hspace{-7em} \forall i \in \{1, \dotsc, n-1\}, \label{eq:path-triangle} \\
& x_{j,k} + x_{j+1,k-1} \leq x_{j+1,k} + x_{j,k-1} \nonumber \\
& \qquad \qquad \forall j,k \in \{0, \dotsc, n\}, \: j < k-2 \label{eq:path-square} \Big \}.
\end{align}

Note that the system \eqref{eq:path-box} - \eqref{eq:path-square} consists precisely of those trivial and square inequalities which we have shown to define facets of $\Xi_P$. We first prove that $\Theta^{\mathrm{Path}}_P$ indeed yields an outer relaxation of $\Xi_P$.

\begin{lem} \label{lem:path-polytope}
It holds that $\Xi_P \subseteq \Theta^{\mathrm{Path}}_P \subseteq \Theta^1_P$.
\end{lem}

As our main result, we prove that $\Theta^{\mathrm{Path}}_P$ is in fact a complete description of $\Xi_P$. 
To derive this, we utilize the notion of total dual integrality. For an extensive reference on this subject we refer the reader to \citep{Schrijver1986}.

\begin{defn}
A system of linear inequalities $Ax \leq b$ with $A \in \Q^{k \times m}, \; b \in \Q^k$ is \emph{totally dual integral} (TDI) if for any $c \in \Z^m$ such that the linear program $\max \{ c^\top x \mid Ax \leq b \}$ is feasible and bounded, there exists an integral optimal dual solution.
\end{defn}

Total dual integrality is an important concept in polyhedral geometry as it serves as a sufficient condition on the integrality of polyhedra according to the following fact.

\begin{fact}[\citep{Edmonds1977}]
\label{fact:tdi-integral}
If $Ax \leq b$ is totally dual integral and $b$ is integral, then the polytope defined by $Ax \leq b$ is integral.
\end{fact}

\begin{thm} 
\label{thm:path-tdi}
The system \eqref{eq:path-box} - \eqref{eq:path-square} is totally dual integral.
\end{thm}

\begin{cor} 
\label{cor:complete-description}
It holds that $\Xi_P = \Theta^{\mathrm{Path}}_P$.
\end{cor}

\begin{proof}
This is immediate from Lemma \ref{lem:path-polytope}, Fact  \ref{fact:tdi-integral} and Theorem \ref{thm:path-tdi}.
\end{proof}

\begin{rem}
The constraint matrix corresponding to the system \eqref{eq:path-box} - \eqref{eq:path-square} is in general \emph{not} totally unimodular. A minimal example is the path of length $4$.
\end{rem}

The set partitioning representation of the path partition problem w.r.t.\ $P$ admits a quadratic number of variables and a linear number of constraints (opposed to a quadratic number of constraints in the description of $\Xi_P$). This representation corresponds to the dual program of the last problem in the proof of Theorem \ref{thm:path-tdi}, cf.\ appendix \ref{sec:proofs}. Here, the integrality constraint need not be enforced, since the constraint matrix is totally unimodular.


\section{Conclusion}

We have characterized all trivial facets of the lifted multicut polytope for trees. Additionally, we provide a tighter relaxation compared to the standard linear relaxation by including additional classes of facets. The described facets provide a complete totally dual integral (TDI) description of the associated lifted multicut polytope for paths. This result relates the geometry of this problem to the combinatorial properties of alternative formulations such as the sequential set partitioning problem.

\bibliographystyle{plainnat}
\bibliography{0000}

\begin{thebibliography}{16}
\providecommand{\natexlab}[1]{#1}
\providecommand{\url}[1]{\texttt{#1}}
\expandafter\ifx\csname urlstyle\endcsname\relax
  \providecommand{\doi}[1]{doi: #1}\else
  \providecommand{\doi}{doi: \begingroup \urlstyle{rm}\Url}\fi

\bibitem[Bansal et~al.(2004)Bansal, Blum, and Chawla]{Bansal2004}
Nikhil Bansal, Avrim Blum, and Shuchi Chawla.
\newblock Correlation clustering.
\newblock \emph{Machine Learning}, 56\penalty0 (1--3):\penalty0 89--113, 2004.
\newblock \doi{10.1023/B:MACH.0000033116.57574.95}.

\bibitem[Boros and Gruber(2012)]{Boros2012}
Endre Boros and Aritanan Gruber.
\newblock On quadratization of pseudo-boolean functions.
\newblock In \emph{International Symposium on Artificial Intelligence and
  Mathematics, {ISAIM} 2012, Fort Lauderdale, Florida, USA, January 9-11,
  2012}, 2012.

\bibitem[Boros and Hammer(2002)]{Boros2002}
Endre Boros and Peter~L. Hammer.
\newblock Pseudo-boolean optimization.
\newblock \emph{Discrete Applied Mathematics}, 123\penalty0 (1–3):\penalty0
  155 -- 225, 2002.
\newblock \doi{https://doi.org/10.1016/S0166-218X(01)00341-9}.

\bibitem[Chopra and Rao(1993)]{Chopra1993}
Sunil Chopra and M.R. Rao.
\newblock The partition problem.
\newblock \emph{Mathematical Programming}, 59\penalty0 (1--3):\penalty0
  87--115, 1993.
\newblock \doi{10.1007/BF01581239}.

\bibitem[Demaine et~al.(2006)Demaine, Emanuel, Fiat, and
  Immorlica]{Demaine2006}
Erik~D. Demaine, Dotan Emanuel, Amos Fiat, and Nicole Immorlica.
\newblock Correlation clustering in general weighted graphs.
\newblock \emph{Theoretical Computer Science}, 361\penalty0 (2--3):\penalty0
  172--187, 2006.
\newblock \doi{10.1016/j.tcs.2006.05.008}.

\bibitem[Deza and Laurent(1997)]{Deza1997}
Michel~Marie Deza and Monique Laurent.
\newblock \emph{Geometry of Cuts and Metrics}.
\newblock Springer, 1997.

\bibitem[Edmonds and Giles(1977)]{Edmonds1977}
Jack Edmonds and Frederick~Richard Giles.
\newblock A min-max relation for submodular functions on graphs.
\newblock \emph{Annals of Discrete Mathematics}, 1:\penalty0 185 -- 204, 1977.
\newblock \doi{http://dx.doi.org/10.1016/S0167-5060(08)70734-9}.

\bibitem[Gr{\"o}tschel and Wakabayashi(1989)]{Groetschel1989}
Martin Gr{\"o}tschel and Yoshiko Wakabayashi.
\newblock A cutting plane algorithm for a clustering problem.
\newblock \emph{Mathematical Programming}, 45\penalty0 (1):\penalty0 59--96,
  1989.
\newblock \doi{10.1007/BF01589097}.

\bibitem[Hammer et~al.(1984)Hammer, Hansen, and Simeone]{Hammer1984}
Peter~L. Hammer, Pierre Hansen, and Bruno Simeone.
\newblock Roof duality, complementation and persistency in quadratic 0--1
  optimization.
\newblock \emph{Mathematical Programming}, 28\penalty0 (2):\penalty0 121--155,
  1984.
\newblock \doi{10.1007/BF02612354}.

\bibitem[Hor\v{n}\'akov\'a et~al.(2017)Hor\v{n}\'akov\'a, Lange, and
  Andres]{hornakova-2017}
Andrea Hor\v{n}\'akov\'a, Jan-Hendrik Lange, and Bjoern Andres.
\newblock Analysis and optimization of graph decompositions by lifted
  multicuts.
\newblock In \emph{ICML}, 2017.

\bibitem[Joseph and Bryson(1997)]{Joseph1997}
Anito Joseph and Noel Bryson.
\newblock Partitioning of sequentially ordered systems using linear
  programming.
\newblock \emph{Computers \& Operations Research}, 24\penalty0 (7):\penalty0
  679 -- 686, 1997.
\newblock \doi{http://dx.doi.org/10.1016/S0305-0548(96)00070-6}.

\bibitem[Kahl and Strandmark(2012)]{Kahl2012}
Fredrik Kahl and Petter Strandmark.
\newblock Generalized roof duality.
\newblock \emph{Discrete Applied Mathematics}, 160\penalty0 (16-17):\penalty0
  2419--2434, 2012.
\newblock \doi{10.1016/j.dam.2012.06.009}.

\bibitem[Kernighan(1971)]{Kernighan1971}
Brian~W. Kernighan.
\newblock Optimal sequential partitions of graphs.
\newblock \emph{J. ACM}, 18\penalty0 (1):\penalty0 34--40, January 1971.
\newblock ISSN 0004-5411.
\newblock \doi{10.1145/321623.321627}.

\bibitem[Padberg(1989)]{Padberg1989}
Manfred Padberg.
\newblock The boolean quadric polytope: Some characteristics, facets and
  relatives.
\newblock \emph{Mathematical Programming}, 45\penalty0 (1):\penalty0 139--172,
  1989.
\newblock \doi{10.1007/BF01589101}.

\bibitem[Rother et~al.(2007)Rother, Kolmogorov, Lempitsky, and
  Szummer]{Rother2007}
Carsten Rother, Vladimir Kolmogorov, Victor~S. Lempitsky, and Martin Szummer.
\newblock Optimizing binary mrfs via extended roof duality.
\newblock In \emph{CVPR}, 2007.

\bibitem[Schrijver(1986)]{Schrijver1986}
Alexander Schrijver.
\newblock \emph{Theory of Linear and Integer Programming}.
\newblock John Wiley \& Sons, Inc., New York, NY, USA, 1986.
\newblock ISBN 0-471-90854-1.

\end{thebibliography}


\appendix

\section{Proofs for Section \ref{sec:polyhedral}}
\label{sec:proofs}

\paragraph{Proof of Lemma \ref{lem:tree-polytope}.}

We show first that $\Xi_T \subseteq \Theta^1_T$. For this purpose, let $x \in \Xi_T$. If $x$ violates \eqref{eq:tree-path}, then $x_{uv} = 1$ and $x_{u,\vec{u}(v)} = x_{\vec{u}(v),v} = 0$. This contradicts the fact that $x$ satisfies all cut inequalities w.r.t.\ $\vec{u}(v),v$ and the path inequality w.r.t.\ $u,v$. If $x$ violates \eqref{eq:tree-cut}, then $x_{\vec{u}(v),v} = 1$ and $x_{uv} = 0$. This contradicts the fact that $x$ satisfies all cut inequalities w.r.t.\ $u,v$ and the path inequality w.r.t.\ $\vec{u}(v),v$. Hence, $x$ satisfies \eqref{eq:tree-path}-\eqref{eq:tree-cut} and therefore $x \in \Theta^1_T$.

Now, we show that $\Theta^1_T \subseteq \Theta^0_T$. Let $x \in \Theta^1_T$. We need to show that $x$ satisfies all path and cut inequalities. Let $u,v \in V$ with $\dist(u,v) \geq 2$. We proceed by induction on $\dist(u,v)$. If $\dist(u,v) = 2$, then the path inequality is directly given by \eqref{eq:tree-path}, while the two cut inequalities are given by \eqref{eq:tree-cut} for the two possible orderings of $u$ and $v$. If $\dist(u,v) > 2$, then the path inequality is obtained from \eqref{eq:tree-path} and the induction hypothesis for the pair $\vec{u}(v),v$, since $\dist(\vec{u}(v),v) = \dist(u,v) - 1$. Similarly, for any edge $e$ on the path from $u$ to $v$, we obtain the cut inequality w.r.t.\ $e$ by using the induction hypothesis and \eqref{eq:tree-cut} such that (w.l.o.g.) $e$ is on the path from $\vec{u}(v)$ to $v$. Hence, $x \in \Theta^0_T$. 

\paragraph{Proof of Lemma \ref{lem:tree-path-facets}.}

First, suppose $\dist(u,v) = 2$. Then $P_{uv}$ is a path of length $2$ and thus chordless in the complete graph on $V$. Hence, facet-definingness follows directly from \citep[Thm 4]{hornakova-2017}. Now, suppose $\dist(u,v) > 2$ and let $x \in \Xi_T$ be such that \eqref{eq:tree-path} is satisfied with equality. We show that this implies
\begin{align}
x_{uv} + x_{\vec{u}(v),\vec{v}(u)} = x_{u,\vec{v}(u)} + x_{\vec{u}(v),v}. \label{eq:square-equality}
\end{align}
Then the face of $\Xi_T$ induced by \eqref{eq:tree-path} has dimension at most $m - 2$ and hence cannot be a facet. In order to check that \eqref{eq:square-equality} holds, we distinguish the following three cases. If $x_{uv} = x_{u,\vec{u}(v)} = x_{\vec{u}(v),v}$, then all terms in \eqref{eq:square-equality} vanish. If $x_{uv} = x_{u,\vec{u}(v)} = 1$ and $x_{\vec{u}(v),v} = 0$, then $x_{\vec{u}(v),\vec{v}(u)} = 0$ and $x_{u,\vec{v}(u)} = 1$, so \eqref{eq:square-equality} holds. Finally, if $x_{uv} = x_{\vec{u}(v),v} = 1$ and $x_{u,\vec{u}(v)} = 0$, then \eqref{eq:square-equality} holds because $x_{\vec{u}(v),\vec{v}(u)} = x_{u,\vec{v}(u)}$ by contraction of the edge $u,\vec{u}(v)$.

\paragraph{Proof of Lemma \ref{lem:tree-cut-facets}.}

First, suppose $v$ is not a leaf of $T$ and let $x \in \Xi_T$ be such that \eqref{eq:tree-cut} is satisfied with equality. Since $v$ is not a leaf, there exists a neighbor $w \in V$ of $v$ such that $P_{\vec{u}(v),v}$ is a subpath of $P_{\vec{u}(v),w}$. We show that $x$ additionally satisfies the equality
\begin{align}
x_{uw} = x_{\vec{u}(v),w} \label{eq:tree-cut-equality}
\end{align}
and thus the face of $\Xi_T$ induced by \eqref{eq:tree-cut} cannot be a facet. There are two possible cases: Either $x_{uv} = x_{\vec{u}(v),v} = 1$, then $x_{uw} = x_{\vec{u}(v),w} = 1$ as well, or $x_{uv} = x_{\vec{u}(v),v} = 0$, then $x_{uw} = x_{vw} = x_{\vec{u}(v),w}$ by contraction of the path $P_{uv}$, so \eqref{eq:tree-cut-equality} holds.

Now, suppose $v$ is a leaf of $T$ and let $\Sigma$ be the face of $\Xi_T$ induced by \eqref{eq:tree-cut}. We need to prove that $\Sigma$ has dimension $m-1$. This can be done explicitly by showing that we can construct $m-1$ distinct indicator vectors $\chi_{w,w'}$ for $w, w' \in V$ as linear combinations of elements from the set $S = \{x \in \Xi_T \mid x_{uv} = x_{\vec{u}(v),v}\}$.
In fact this construction is analogous to the one used in the proof of Theorem 7 in \citep{hornakova-2017} where the authors derive the dimension of the lifted multicut polytope $\dim(\Xi_{GG'}) = \lvert E' \rvert$. The difference here is that the vector $\1 - \chi_{uv} \notin S$, so we have to distinguish between $x \in S$ with $x_{uv} = x_{\vec{u}(v),v} = 0$ and $x \in S$ with $x_{uv} = x_{\vec{u}(v),v} = 1$ in order to show $\chi_{u,\vec{u}(v)} \in \lin(S)$ and then $\chi_e \in \lin(S)$ for all other $e \neq  \{u,v\}, \{\vec{u}(v),v\}$.

\paragraph{Proof of Lemma \ref{lem:tree-box-facets}.}

We apply the more general characterization given in \citep{hornakova-2017} to our special case. The nodes $u,v \in V$ are a pair of $w$-$w'$-cut-vertices for some vertices $w,w' \in V$ (with at least one being different from $u$ and $v$) if and only if $u$ or $v$ is not a leaf of $V$. Thus, the claim follows from \citep[Thm 8]{hornakova-2017}. The second assertion follows from \citep[Thm 9]{hornakova-2017} and the fact that we lift to the complete graph on $V$.

\paragraph{Proof of Lemma \ref{lem:tree-square-valid}.}

Let $x \in \Xi_T$ and suppose that either $x_{u,\vec{v}(u)} = 0$ or $x_{\vec{u}(v),v} = 0$ for some $u,v \in V$ with $\dist(u,v) \geq 3$. Then, since $x$ satisfies all cut inequalities w.r.t.\ $u,\vec{v}(u)$, respectively $\vec{u}(v),v$, and the path inequality w.r.t.\ $\vec{u}(v), \vec{v}(u)$, it must hold that $x_{\vec{u}(v),\vec{v}(u)} = 0$. Moreover, if even $x_{u,\vec{v}(u)} = 0 = x_{\vec{u}(v),v}$, then, by the same reasoning, we have $x_{uv} = 0$ as well. Hence, $x$ satisfies \eqref{eq:tree-square}.

\paragraph{Proof of Lemma \ref{lem:tree-square-facets}.}

Let $\Sigma$ be the face of $\Xi_T$ induced by \eqref{eq:tree-square}. We need to prove that $\Sigma$ has dimension $m-1$, which can be done explicitly by showing that we can construct $m-1$ distinct indicator vectors $\chi_{w,w'}$ for $w, w' \in V$ as linear combinations of elements from the set $S = \{x \in \Xi_T \mid x_{uv} + x_{\vec{u}(v),\vec{v}(u)} = x_{u,\vec{v}(u)} + x_{\vec{u}(v),v}\}$.
Again, this construction is very technical and analogous to the proof of Theorem 7 in \citep{hornakova-2017} about the dimension of the lifted multicut polytope.

\paragraph{Proof of Lemma \ref{lem:path-polytope}.}

First, we show that $\Xi_P \subseteq \Theta^{\mathrm{Path}}_P$. Let $x \in \Xi_P$, then $x$ satisfies \eqref{eq:path-box} and \eqref{eq:path-triangle} by definition. Suppose $x$ violates \eqref{eq:path-cut-right}, then $x_{in} = 1$ and $x_{i-1,n} = 0$. This contradicts the fact that $x$ satisfies all cut inequalities w.r.t.\ $i-1,n$ and the path inequality w.r.t.\ $i,n$. So, $x$ must satisfy \eqref{eq:path-cut-right} and, by symmetry, also \eqref{eq:path-cut-left}. It follows from Lemma \ref{lem:tree-square-valid} that $x$ satisfies \eqref{eq:path-square} as well and thus $x \in \Theta^{\mathrm{Path}}_P$.

Next, we prove that $\Theta^{\mathrm{Path}}_P \subseteq \Theta^1_P$. To this end, let $x \in \Theta^{\mathrm{Path}}_P$.
We show that $x$ satisfies all inequalities \eqref{eq:tree-cut}. Let $u,v \in V$ with $u < v - 1$. We need to prove that both $x_{u+1,v} \leq x_{uv}$ and $x_{u,v-1} \leq x_{uv}$ hold. For reasons of symmetry, it suffices to show only $x_{u+1,v} \leq x_{uv}$. We proceed by induction on the distance of $u$ from $n$. If $v = n$, then $x_{u+1,n} \leq x_{un}$ is given by \eqref{eq:path-cut-right}. Otherwise, we use \eqref{eq:path-square} for $j=u$ and $k=v+1$ and the induction hypothesis on $v+1$:
\begin{align*}
x_{uv} + x_{u+v,v+1} & \geq x_{u+1,v} + x_{u,v+1} \\
& \geq x_{u+1,v} + x_{u+1,v+1} \\
\implies x_{uv} & \geq x_{u+1,v}.
\end{align*}  
It remains to show that $x$ satisfies all inequalities \eqref{eq:tree-path}. Let $u,v \in V$ with $u < v - 1$. We proceed by induction on $\dist(u,v) = u-v$. If $\dist(u,v) = 2$, then \eqref{eq:tree-path} is given by \eqref{eq:path-triangle}. If $\dist(u,v) > 2$, then we use \eqref{eq:path-square} for $j=u$ and $k=v$ as well as the induction hypothesis on $u,v-1$, which have distance $\dist(u,v) - 1$:
\begin{align*}
x_{uv} + x_{u+1,v-1} & \leq x_{u+1,v} + x_{u,v-1} \\
&  \leq x_{u+1,v} + x_{u,u+1} + x_{u+1,v-1} \\
\implies  x_{uv} & \leq x_{u,u+1} + x_{u+1,v}.
\end{align*}
Hence, $x \in \Theta^1_P$, which concludes the proof.

\paragraph{Proof of Theorem \ref{thm:path-tdi}}

We rewrite system \eqref{eq:path-box} - \eqref{eq:path-square} more compactly in the following way. Introduce two artificial nodes $-\infty$ and $\infty$ where we associate $-\infty$ to any index less than $0$ and $\infty$ to any index greater than $n$. Moreover, we introduce variables $x_{ii}$ for all $0 \leq i \leq n$ as well as $x_{-\infty,i}$ and $x_{i, \infty}$ for all $1 \leq i \leq n-1$ and finally $x_{-\infty, \infty}$. Now, the system \eqref{eq:path-box} - \eqref{eq:path-square} is equivalent to the system
\begin{align}
& x_{j,k} + x_{j+1,k-1} \leq x_{j+1,k} + x_{j,k-1} \nonumber \\
& \qquad \qquad \forall j,k \in \{-\infty, 0, \dotsc, n, \infty \}, \: j \leq k-2 \label{eq:path-4cycle}
\end{align}
given the additional equality constraints
\begin{align}
x_{ii} & = 0 \quad \forall 0 \leq i \leq n, \label{eq:path-vertices} \\
x_{-\infty, i} & = 1 \quad \forall 1 \leq i \leq n-1, \\
x_{i, \infty} & = 1 \quad \forall 1 \leq i \leq n-1, \\
x_{-\infty, \infty} & = 1. \label{eq:path-inf-cut}
\end{align}
Let the system defined by \eqref{eq:path-4cycle} - \eqref{eq:path-inf-cut} be represented in matrix form as $Ax \leq a, \; Bx = b$. Note that $\Theta^{\mathrm{Path}}_P$ is non-empty and bounded. Thus, to establish TDI, we need to show that for any $c \in \Z^{m + 3n}$ the dual program
\begin{align}
\min \{ a^\top y + b^\top z \mid A^\top y + B^\top z = c, y \geq 0\} \label{eq:dual-prog}
\end{align}
has an integral optimal solution. Here, the $y$ variables, indexed by $j,k$, correspond to the inequalities \eqref{eq:path-4cycle} and the $z$ variables, indexed by pairs of $i, -\infty$ and $\infty$, correspond to the equations \eqref{eq:path-vertices}  - \eqref{eq:path-inf-cut}. Then, the equation system $A^\top y + B^\top z = c$ translates to
\begin{align}
y_{i-1,i+1} + z_{i,i} & = c_{i,i} \label{eq:dual-d1} \\
y_{i-1,i+2} - y_{i-1,i+1} - y_{i,i+2} & = c_{i,i+1} \label{eq:dual-d2} \\
y_{i-1,\ell+1} - y_{i-1,\ell} - y_{i,\ell+1} + y_{i,\ell} & = c_{i,\ell} \label{eq:dual-d3} \\
- y_{-\infty,i+1} + y_{-\infty,i} + z_{-\infty,i} & = c_{-\infty,i} \label{eq:dual-left} \\
- y_{i-1,\infty} + y_{i,\infty} + z_{i,\infty} & = c_{i,\infty} \label{eq:dual-right} \\
\phantom{\ +} y_{-\infty,\infty} + z_{-\infty,\infty} & = c_{-\infty,\infty}, \label{eq:dual-inf}
\end{align}
where \eqref{eq:dual-d1} - \eqref{eq:dual-d3} hold for all $0 \leq i < i+1 < \ell \leq n$ and \eqref{eq:dual-left}, \eqref{eq:dual-right} hold for all $1 \leq i \leq n-1$. Observe that \eqref{eq:dual-d1} includes only $y$ variables with indices of distance $2$, \eqref{eq:dual-d2} couples $y$ variables of distance $3$ with those of distance $2$ and \eqref{eq:dual-d3} couples the remaining $y$ variables of any distance $d > 3$ with those of distance $d-1$ and $d-2$. Hence, any choice of values for the free variables $z_{ii}$ completely determines all $y$ variables. This means we can eliminate $y$ and reformulate the dual program entirely in terms of the $z$ variables, as follows. It holds that
\begin{align*}
0 \leq y_{i-1,i+1} & = c_{i,i} - z_{i,i} \qquad \forall 0 \leq i \leq n, \\
0 \leq y_{i-1,i+2} & = c_{i,i+1} + c_{i,i} + c_{i+1,i+1} - z_{i,i} - z_{i+1,i+1} \\
& \qquad \qquad \forall 0 \leq i < i+1 \leq n
\end{align*}
and thus, by \eqref{eq:dual-d3},
\begin{align*}
0 \leq y_{i-1,\ell+1} = \sum_{i \leq j \leq k \leq \ell} c_{j,k} - \sum_{k=i}^{\ell} z_{k,k}
\end{align*}
for all $0 \leq i \leq \ell \leq n$. Substituting the $y$ variables in \eqref{eq:dual-left} - \eqref{eq:dual-inf} yields the following equivalent formulation of \eqref{eq:dual-prog}:
\begin{align}
\min \quad & z_{-\infty,\infty} + \sum_{i=0}^n z_{-\infty,i} + z_{i,\infty} \label{eq:dual-red} \\
\text{s.t.} \quad & \sum_{k=i}^\ell z_{k,k} \leq \sum_{i \leq j \leq k \leq \ell} c_{j,k} \qquad \forall 0 \leq i \leq \ell \leq n \nonumber \\
& z_{-\infty,i} + z_{i,i} = c_{-\infty,i} + \sum_{0 \leq j \leq i} c_{j,i} \quad \forall 1 \leq i \leq n-1 \nonumber \\
& z_{i,\infty} + z_{i,i} = c_{i,\infty} + \sum_{i \leq k \leq n} c_{i,k} \qquad \forall 1 \leq i \leq n-1 \nonumber \\
& z_{-\infty,\infty} - \sum_{k=1}^n z_{k,k} = c_{-\infty,\infty} - \sum_{0 \leq j \leq k \leq n} c_{j,k}. \nonumber 
\end{align}
The variables $z_{-\infty,i}, \; z_{i,\infty}$ and $z_{-\infty,\infty}$ occur only in a single equation each. Furthermore, the matrix corresponding to the inequality constraints satisfies the ``consecutive-ones'' property. Therefore, the constraint matrix of the whole system is totally unimodular, which concludes the proof.

\paragraph{Path Partition as Sequential Set Partitioning}

For each $0 \leq i \leq \ell \leq n$, let 
\begin{align*}
d_{i,\ell} = \sum_{0 \leq i \leq j \leq k \leq n} c_{j,k},
\end{align*}
then taking the dual of problem \eqref{eq:dual-red} and simplifying yields the problem
\begin{align}
\min \quad & d^\top \lambda \label{eq:ssp} \\
\text{s.t.} \quad & \sum_{0 \leq i \leq k \leq \ell \leq n} \lambda_{i,\ell} = 1, \quad \forall 0 \leq k \leq n \nonumber \\
& \lambda \geq 0. \nonumber 
\end{align}
Each variable $\lambda_{i,\ell}$ corresponds to the component containing nodes $i$ to $\ell$. Problem \eqref{eq:ssp} is precisely the sequential set partitioning formulation of the \pathproblem{} as used by \citet{Joseph1997}.

\section{Lifted Multicuts of General Graphs}
\label{sec:lifted-multicuts}

For any connected graph $G = (V,E)$,
any supergraph $G' = (V, E')$ with $E' = E \cupdot F$ 
and any $c: E' \to \mathbb{R}$
the instance of the \emph{minimum cost lifted multicut problem} w.r.t.~$G$, $G'$ and $c$ is the binary linear program
\begin{align}
\min_{x \in X_{GG'}} \sum_{e \in E'} c_e \, x_e \label{eq:lmp-problem}
\end{align}
with
\begin{align}
X_{GG'} & = \Big \{ x \in \{0,1\}^{E'} \; \big | \nonumber \\
& \forall C \in \text{cycles}(G) \, \forall \hat{e} \in C : x_{\hat{e}} \leq \sum_{e \in C \setminus \{\hat{e}\}} x_e, \label{eq:lmp-cycle} \\
& \forall uv \in F \, \forall P \in uv\text{-paths}(G) : x_{uv} \leq \sum_{e \in P} x_e, \label{eq:lmp-path} \\
& \forall uv \in F \, \forall C \in uv\text{-cuts}(G) : 1 - x_{uv} \leq \sum_{e \in C} 1 - x_e \label{eq:lmp-cut} \Big \}. 
\end{align}
The inequalities \eqref{eq:lmp-cycle}, \eqref{eq:lmp-path} and \eqref{eq:lmp-cut} are referred to as \emph{cycle}, \emph{path} and \emph{cut} inequalities, respectively.
The convex hull 
\begin{align}
\Xi_{GG'} = \conv X_{GG'} \label{eq:lifted-multicut-polytope}
\end{align}
of $X_{GG'}$ in $\mathbb{R}^{E'}$ is called the \emph{lifted multicut polytope} w.r.t.~$G$ and $G'$.
It was shown by \citet{hornakova-2017} that $\Xi_{GG'}$ has dimension $\lvert E' \rvert$.

For any $x \in X_{GG'}$, the set $M := \{e \in E \,|\, x_e = 1\}$ of those edges of the graph $G$ that are labeled 1 is a so-called \emph{multicut} of $G$. 
Hence, there exists a decomposition of $G$ such that $M$ is precisely the set of those edges that span across distinct components.
In addition, the set $M' := \{e \in E' \,|\, x_e = 1\}$ of those edges $\{u,v\} = e \in E'$ of the graph $G'$ that are labeled 1 is a multicut of $G'$ lifted from the multicut $M$ of $G$.
A multicut $M'$ of $G'$ lifted from $G$ makes explicit for all pairs $\{u,v\} \in E'$ of nodes (not only those neighboring in $G$) whether $u$ and $v$ are in the same component in the decomposition of $G$.

\end{document}